\documentclass[aps,pra,superscriptaddress]{revtex4}

\usepackage{epsfig}
\usepackage{graphicx}
\usepackage{amssymb}
\usepackage{amsmath}
\usepackage{amsthm}
\usepackage{hyperref}
\usepackage{bm}

\newcommand{\nn}{\nonumber \\}
\newcommand{\bra}[1]{\langle{#1}|}
\newcommand{\ket}[1]{|{#1}\rangle}
\newcommand{\braket}[2]{\langle{#1}|{#2}\rangle}

\newcommand{\poly}{\mathop{\mathrm{poly}}}
\newcommand{\diag}{\mathop{\mathrm{diag}}}
\newcommand{\matr}{{\cal A}}
\newcommand{\dt}{h}
\newcommand{\s}{z}

\def \A {{\cal A}}
\def \B {{\cal B}}

\newtheorem{theorem}{Theorem}
\newtheorem{lemma}[theorem]{Lemma}

\newtheorem{definition}[theorem]{Definition}

\begin{document}

\title{High-order quantum algorithm for solving linear differential equations}

\author{Dominic W. Berry}
\affiliation{Department of Physics and Astronomy, Macquarie University, Sydney, NSW 2109, Australia}
\affiliation{Institute for Quantum Computing, University of Waterloo, Waterloo, Ontario N2L 3G1, Canada}

\begin{abstract}
Linear differential equations are ubiquitous in science and engineering.
Quantum computers can simulate quantum systems, which are described by a restricted type of linear differential equations.
Here we extend quantum simulation algorithms to general inhomogeneous sparse linear differential equations, which describe many classical physical systems.
We examine the use of high-order methods to improve the efficiency.
These provide scaling close to $\Delta t^{2}$ in the evolution time $\Delta t$.
As with other algorithms of this type, the solution is encoded in amplitudes of the quantum state, and it is possible to extract global features of the solution.
\end{abstract}
\date{\today}

\maketitle

\section{Introduction}
Differential equations are used for an enormous variety of applications, including industrial design and weather prediction.
In fact, many of the main applications of supercomputers are in the form of large systems of differential equations \cite{super}.
Therefore quantum algorithms for solving differential equations would be extraordinarily valuable.
A quantum algorithm for differential equations was proposed in Ref.\ \cite{Leyton08}, but that algorithm had very poor scaling in the time.
The complexity of the simulation scaled exponentially in the number of time-steps over which to perform the simulation.

The algorithm in Ref.\ \cite{Leyton08} may have been overly ambitious, because it aimed to solve nonlinear differential equations.
A more natural application for quantum computers is \emph{linear} differential equations.
This is because quantum mechanics is described by linear differential equations.
We find that, when we restrict to linear differential equations, it is possible to obtain an algorithm that is far more efficient than that proposed in Ref.\ \cite{Leyton08}.

We consider sparse systems of first-order linear differential equations.
Using standard techniques, any linear differential equation with higher-order derivatives can be converted to a first-order linear differential equation with larger dimension.
A first-order ordinary differential equation may be written as
\begin{equation}
\bm{\dot x}(t) = A(t)\bm{x}(t) + \bm{b}(t),
\end{equation}
where $\bm{x}$ and $\bm{b}$ are $N_x$-component vectors, and $A$ is an $N_x\times N_x$ matrix (which we take to be sparse).
Classically, the complexity of solving the differential equation must be at least linear in $N_x$.
The goal of the quantum algorithm is to solve the differential equation in time $O(\poly\log N_x)$.

Quantum mechanics is described by differential equations of this form, except they are homogeneous ($\bm{b}(t)=0$), and $A(t)=iH(t)$, where $H(t)$ is Hermitian.
For physical systems where the Hamiltonian is a sum of interaction Hamiltonians, the Hamiltonian is sparse.
The solutions in quantum mechanics only include oscillating terms, whereas more general differential equations have solutions that may grow or decay exponentially.
Quantum algorithms for simulating quantum mechanical systems have been extensively studied \cite{Lloyd96,Aharonov03,Childs04,Berry07,Childs08,Berry09,Wiebe10,Berry13,Jordan12}.

Classical physics is described by more general differential equations.
Large sparse systems of ordinary differential equations are produced by discreti{\s}ation of partial differential equations.
Many equations in physics are linear partial differential equations, where the time derivative depends linearly on spatial derivatives and the value of a quantity at some point in physical space. Examples include Stokes equations (for creeping fluid flow), the heat equation, and Maxwell's equations.
Discreti{\s}ation of the partial differential equation on a mesh of points results in a sparse ordinary differential equation with a very large value of $N_x$.

In the case where $A$ and $\bm{b}$ are time independent, then one can find the steady-state solution of the differential equation by solving
\begin{equation}
A \bm{x} = -\bm{b}.
\end{equation}
A quantum algorithm for this problem was given by Harrow, Hassidim and Lloyd (HHL) \cite{Harrow09}, with runtime that is polynomial in the condition number of $A$ and the log of the dimension.
In particular, its runtime is $\tilde O(\log(N)s^4\kappa^2/\epsilon_L)$, where $s$ is the sparseness of the matrix $A$, $\kappa$ is its condition number, $N$ is the dimension, and $\epsilon_L$ is the allowable error.
The $\tilde O$ in that work is taken to mean that arbitrarily small powers are ignored, so scaling as $t^{1+\delta}$ for arbitrarily small $\delta$ is $\tilde O(t)$.
Ambainis has presented an alternative technique using variable time amplitude amplification \cite{Ambainis10,Ambainis10b,Ambainis12}.
In this work we primarily concentrate on the HHL algorithm, but also examine the technique of \cite{Ambainis10,Ambainis10b,Ambainis12} in Sec.~\ref{sec:var}.

We show how to solve the full time evolution under sparse linear differential equations on a quantum computer.
The goal of this work is to not only provide an algorithm that is efficient in that it scales as the log of the dimension, but to provide efficient scaling in the evolution time $\Delta t$.
To improve the efficiency of the algorithm, we consider high-order linear multistep methods.
Given reasonable conditions on the problem, the complexity scales as $\tilde O(\Delta t^{2})$.
In comparison, high-order methods for Hamiltonian simulation yield scaling close to linear in $\Delta t$ \cite{Berry07};
this is close to optimal, due to the no--fast-forwarding theorem \cite{Berry07}.
There is an interesting open question whether the scaling in $\Delta t$ for differential equations can be further improved.
As in Ref.~\cite{Harrow09}, the final solution is encoded in the amplitudes of a state vector.
The solution is not given explicitly, but in many cases of interest it is expected that useful global information can be extracted from the state.
Methods to achieve this are discussed in Ref.~\cite{Clader}.

\section{Lie-Trotter formula approach}

Before explaining the main approach, we first describe an approach using Lie-Trotter formulae, and the drawback to that approach.
This will not be described rigorously, because it is not the main proposal for solving differential equations.

The homogeneous case, where $\bm{b}=0$, is analogous to Hamiltonian evolution. If $A$ is antiHermitian, then we can take $A=iH$, where $H$ is a Hermitian Hamiltonian.
Evolution under this Hamiltonian can be solved by methods considered in previous work \cite{Berry07,Berry09}.
Another case that can be considered is where $A$ is Hermitian.
In this case, the eigenvalues of $A$ are real, and $A$ can be diagonali{\s}ed in the form $A=V D V^{-1}$, where $D$ is a real diagonal matrix and $V$ is unitary.
The formal solution is then, for $A$ independent of time, $\bm{x}(t)=V e^{D(t-t_0)} V^{-1} \bm{x}(t_0)$.

The differential equation can be solved using a similar method to that used in Ref.\ \cite{Harrow09}. The value of $\bm{x}$ is encoded in a quantum state proportional to
\begin{equation}
\ket{\bm{x}} := \sum_{j=1}^{N_x} x^{[j]} \ket{j},
\end{equation}
where $\ket{j}$ are computational basis states of the quantum computer, and $x^{[j]}$ are the components of the vector $\bm{x}$.
It is convenient to omit the normali{\s}ation for this state.
The state can be written in a basis corresponding to the eigenvectors of $A$:
\begin{equation}
\ket{\bm{x}} = \sum_{j} \psi_j \ket{\lambda_j}.
\end{equation}
Using methods for Hamiltonian simulation, $iA$ can be simulated.
By using phase estimation, if the state is an eigenstate $\ket{\lambda_j}$, then the eigenvalue $\lambda_j$ can be determined.
Given maximum eigenvalue $\lambda_{\rm max}$, we would change the amplitude by a factor of $e^{(t-t_0)(\lambda_j-\lambda_{\rm max})}$.
See Ref.\ \cite{Harrow09} for the method of changing the amplitude. If this is done coherently, then the final state will encode $\bm{x}(t)$.

For more general differential equations, $A$ will be neither Hermitian nor antiHermitian. In this case, one can break $A$ up into Hermitian ($A_{H}$) and antiHermitian ($A_{aH}$) components. The evolution under each of these components can be simulated individually, and the overall evolution simulated by combining these evolutions via the Lie-Trotter formula. The drawback to this approach is that it appears to give a complexity that increases exponentially with the time interval $\Delta t = t-t_0$ (though the complexity is still greatly improved over Ref.\ \cite{Leyton08}).

If $A$ were just Hermitian, then the eigenvector (or eigenspace) corresponding to the largest eigenvalue would not decay, and the system would end up in that state.
Therefore the amplitude would not drop below the amplitude on the eigenspace corresponding to the largest eigenvalue.
That is not the case when $A$ is a more general matrix, because usually the maximum real part of an eigenvalue of $A$ will be strictly less than the maximum eigenvalue of $A_H$.
The amplitude must therefore decay exponentially, because we must allow for the maximum eigenvalue of $A_H$ in simulating evolution under $A_H$.

The result of this is that the complexity of the simulation will scale exponentially in the time that the differential equation needs to be simulated over, $\Delta t$.
The scaling will be considerably improved over that in Ref.\ \cite{Leyton08}, but it is desirable to obtain scaling that is polynomial in $\Delta t$.
Another drawback is that this approach does not enable simulation of inhomogeneous differential equations.

\section{Linear systems approach}

To avoid this problem we propose an approach based on the algorithm for solving linear systems from Ref.\ \cite{Harrow09}.
The trick is to use a ``Feynman's clock'', where an additional register encodes the time \cite{feyn82,feyn85}.
Then we encode the solution of the differential equation at all times using just one state.
That is, we wish to obtain the final state proportional to
\begin{equation}
\label{eq:fineq}
\ket{\psi} := \sum_{j=0}^{N_t} \ket{t_j} \ket{\bm{x}_j}.
\end{equation}
The number $N_t$ is the number of time steps, $t_j$ is the time $t_0+j\dt$, where $\dt$ is the time interval in the discreti{\s}ation of the differential equation, $\bm{x}_j$ is the approximation of the value of $\bm{x}$ at time $t_j$, and $\Delta t$ is the total time interval over which the differential equation is to be solved. We use the subscript $j$ to index the vectors, and superscript for components of these vectors.

Once this state has been created, the state encoding the solution at the final time $t_0+\Delta t$ can be approximated by measuring the register encoding the time and getting that time.
Just using this method, the probability of obtaining the final time is small ($1/(N_t+1)$).
To obtain a significant probability of success, one can add times beyond $t_0+\Delta t$ where $\bm{x}$ is constant.
We take $\bm{x}$ to be constant for $t_0+\Delta t$ to $t_0+2\Delta t$, so $N_t=2\Delta t/\dt$.
Then any measurement result for the time in this interval will give the state corresponding to the solution.
By this method, the probability of success can be boosted significantly, without changing the scaling for $N_t$.

To numerically solve differential equations, the simplest method is the Euler method, which discreti{\s}es the differential equation as
\begin{equation}
\frac{\bm{x}_{j+1}-\bm{x}_j}{\dt} = A(t_j) \bm{x}_j +\bm{b}(t_j).
\end{equation}
For times after $t_0+\Delta t$, we set $\bm{x}_{j+1}=\bm{x}_j$ to ensure that $\bm{x}$ is constant.
It is straightforward to encode this method as a linear system
\begin{equation}
\matr\bm{\vec x} = \bm{\vec b}.
\end{equation}
Here $\bm{\vec x}$ is the vector of blocks $\bm{x}_j$, $\bm{\vec b}$ is a vector of the blocks $\bm{b}$ and an initial condition $\bm{x}_{\rm in}$, and $\matr$ is a matrix describing the discretised differential equation.
For $A$ and $\bm{b}$ independent of time a simple example may be expressed as
\begin{equation}
\label{eq:exmat}
\left[ {\begin{array}{*{20}c}
   \openone & 0 & 0 & 0 & 0  \\
   {-(\openone + A\dt)} & \openone & 0 & 0 & 0  \\
   0 & {-(\openone + A\dt)} & \openone & 0 & 0  \\
   0 & 0 & -\openone & \openone & 0  \\
   0 & 0 & 0 & -\openone & \openone  \\
\end{array}} \right]\left[ {\begin{array}{*{20}c}
   {\bm{x}_0 }  \\
   {\bm{x}_1 }  \\
   {\bm{x}_2 }  \\
   {\bm{x}_3 }  \\
   {\bm{x}_4 }  \\
\end{array}} \right] = \left[ {\begin{array}{*{20}c}
   {\bm{x}_{\rm in} }  \\
   \bm{b}\dt  \\
   \bm{b}\dt  \\
   0  \\
   0  \\
\end{array}} \right].
\end{equation}

Each entry of $\matr$ is a block of the dimension of $A$, and each entry of $\bm{\vec x}$ and $\bm{\vec b}$ is a block of the dimension of $\bm{x}$. 
The first row sets the initial value, $\bm{x}_0=\bm{x}_{\rm in}$. The next rows give $\bm{x}_{j+1}-(\bm{x}_j+A \bm{x}_j \dt)=\bm{b}\dt$, corresponding to the discreti{\s}ation of the differential equation via the Euler method.
The final rows indicate equations where $\bm{x}_{j+1}-\bm{x}_j=0$. This is for the times where $\bm{x}$ is required to be constant.

The Euler method yields an error that scales as $O(\dt^2)$ for a single time step.
Therefore, we expect that the error in the total simulation is $O(N_t \dt^2)=O(\Delta t^2/N_t)$.
To achieve error bounded by $\epsilon$, we can take $N_t = O(\Delta t^2/\epsilon)$. 
Here $\epsilon$ is the error in the solution of the differential equation, whereas the error in the solution of linear systems is denoted $\epsilon_L$.
To determine the scaling for solving this system via the HHL algorithm then requires the condition number to be determined.

As a rough estimate of the condition number, $\|\matr^{-1}\|$ can be estimated by considering how large $\bm{\vec x}$ can be for a given $\bm{\vec b}$.
Intuitively, we can consider each component of $\bm{\vec b}$ as an excitation, which can result in an increase in $\bm{\vec x}$ for later times.
Because there are $O(N_t)$ times that can be considered, this means that we expect $\bm{\vec x}$ can be increased by a factor of $O(N_t)$ over $\bm{\vec b}$.
It is easy to see that $\|\matr\|$ is close to $1$, and so the condition number is $O(N_t)$.
This then results in a scaling of the complexity of at least $O(\Delta t^4)$ (this argument will be made rigorous in the following sections).

This scaling in $\Delta t$ is unexpectedly poor, in comparison to Hamiltonian simulation where the scaling is linear (or close to linear) in $\Delta t$.
As in the case of Hamiltonian simulation, it can be expected that higher-order methods can improve the scaling.
In the case of differential equations, the appropriate higher-order methods are linear multistep methods.
In Sec.~\ref{sec:multistep} we introduce multistep methods, then in Sec.~\ref{sec:cond} we explain how to bound the condition number.
Then in Sec.~\ref{sec:fullalg} we explain how the HHL algorithm is applied and the resulting complexity.

\section{Multistep methods}
\label{sec:multistep}
General linear multistep methods have the form \cite{Butcher,Hairer}
\begin{equation}
\label{eq:multi}
\sum_{\ell=0}^{k} \alpha_\ell \bm{x}_{j+\ell} = \dt \sum_{\ell=0}^{k}\beta_\ell [A(t_{j+\ell})  \bm{x}_{j+\ell}+\bm{b}(t_{j+\ell})].
\end{equation}
Multistep methods can be chosen such that the error is of higher order in $\dt$, but there is the problem that the method may not be stable.
To bound the error it is necessary that the eigenvalues of $A(t_j)$ have no positive real part, because otherwise the error can grow exponentially.
In cases where $A(t_j)$ does have an eigenvalue with positive real part, one can simply subtract a multiple of the identity, and rescale the solution.
Even then, when the exact solution of the differential equation is bounded, the solution of the difference equation may be unbounded.

To examine the stability, one defines the generating polynomials
\begin{equation}
\rho(\zeta)=\sum_{j=0}^k \alpha_j \zeta^j, \qquad \sigma(\zeta) = \sum_{j=0}^k \beta_j \zeta^j.
\end{equation}
The stability can be examined via the roots of the equation
\begin{equation}
\label{eq:stpol}
\rho(\zeta)-\mu \sigma(\zeta) = 0.
\end{equation}
One defines the set $S$ by
\begin{equation}
S := \left\{ \mu\in {\mathbb{C}}; \begin{array}{*{20}l}
{{\rm all~roots~} \zeta_j(\mu) {\rm~of~} \eqref{eq:stpol} {\rm~satisfy~} |\zeta_j(\mu)|} \le 1 \\
{{\rm multiple~roots~satisfy~} |\zeta_j(\mu)|< 1} \\ \end{array} \right\}.
\end{equation}
$S$ is called the stability domain or stability region of the multistep method. In addition, if the roots of $\sigma(\zeta)$ all satisfy $|\zeta|\le 1$, and repeated roots satisfy $|\zeta|<1$, then the method is said to be stable at infinity.

A linear multistep method is said to be order $p$ if it introduces local errors $O(\dt^{p+1})$. This means that, if it is applied with exact starting values to the problem $\dot x = t^q$ ($0\le q \le p$), it integrates the problem without error. A linear multistep method has order $p$ if and only if \cite{Butcher}
\begin{equation}
\rho(e^h)-h\sigma(e^h) = O(h^{p+1}).
\end{equation}

A useful property of linear multistep methods is for them to be $A$-stable \cite{Hairer,Dahlquist}.
\begin{definition}
A linear multistep method is called $A$-stable if $S \supset \mathbb{C}^-$, i.e., if
\begin{equation}
{\rm Re}\, \lambda \le 0 \implies \text{numerical solution for } \dot x = \lambda x \text{ is bounded.}
\end{equation}
\end{definition}
This definition means that, if the solution of the differential equation is bounded, then the approximation given by the multistep method is bounded as well. For a scalar differential equation, the multistep method is bounded whenever $\lambda$ is in the left half of the complex plane. The Euler method is $A$-stable, but it is not possible to construct arbitrary order $A$-stable multistep methods. The second Dahlquist barrier is that an $A$-stable multistep method must be of order $p\le 2$ \cite{Hairer,Dahlquist}. As we wish to consider higher-order multistep methods, we relax the condition and require that the linear multistep method is $A(\alpha)$-stable \cite{Hairer,Widlund}.

\begin{definition}
A linear multistep method is $A(\alpha)$-stable, $0<\alpha<\pi/2$, if
\begin{equation}
S \supset S_\alpha = \{\mu ; |\arg(-\mu)| < \alpha, \mu \ne 0 \}.
\end{equation}
\end{definition}
This definition means that, in the case of a scalar differential equation, the multistep method is bounded whenever $\lambda$ is within a wedge in the left half of the complex plane. For a vector differential equation, the eigenvalues of $A$ should be within this wedge. It is known that, for any $\alpha<\pi/2$ and $k\in \mathbb{N}$, there is an $A(\alpha)$-stable linear $k$-step method of order $p=k$ \cite{Grigoreff,Butcher}.

The error in the total solution of the differential equation will be $O(N_t (\Delta t)^{p+1})$. In order to obtain a rigorous result, we speciali{\s}e to the case that $A$ and $b$ are independent of time. The relevant bound is given in Theorem 7.6 in Chapter V of Ref.\ \cite{Hairer}.
\begin{theorem}
\label{thm2}
Suppose a linear multistep method is of order $p$, $A(\alpha)$-stable and stable at infinity. If the matrix $A$ is diagonali{\s}able (i.e.\ there exists a matrix $V$ such that $V^{-1}AV=D=\diag(\lambda_1,\ldots,\lambda_n)$) with eigenvalues satisfying
\begin{equation}
|\arg(-\lambda_i)|\le \alpha \qquad for~i=1,\ldots,N_x,
\end{equation} 
then there exists a constant $M$ (depending only on the method) such that for all $\dt>0$ the global error satisfies
\begin{equation}
\|\bm{x}(t_m)-\bm{x}_m\| \le M \kappa_V \left( \max_{0\le j< k} \| \bm{x}(t_j)-\bm{x}_j \| + \dt^p \int_{t_0}^{t_m} \|\bm{x}^{(p+1)}(\xi)\|d\xi\right),
\end{equation}
where $\kappa_V=\|V\| \cdot \|V^{-1}\|$ is the condition number of $V$.
\end{theorem}

Here the superscript with round brackets denotes repeated derivative. We can use this result to show a lemma on the scaling of the error.

\begin{lemma}
\label{lem:ersca}
Suppose a linear multistep method is of order $p$, $A(\alpha)$-stable and stable at infinity. If the matrix $A$ is diagonali{\s}able (i.e.\ there exists a matrix $V$ such that $V^{-1}AV=D=\diag(\lambda_1,\ldots,\lambda_n)$) with eigenvalues satisfying
\begin{equation}
|\arg(-\lambda_i)|\le \alpha \qquad for~i=1,\ldots,N_x,
\end{equation} 
and $b$ is constant, then the global error satisfies
\begin{equation}
\|\bm{x}(t_m)-\bm{x}_m\| = O\left( \kappa_V^2 (\|\bm{x}_{\rm init}\| + \|\bm{b}\|/\|A\|)\left[ \kappa_V (\dt \|A\|)^2
 + m(\dt \|A\|)^{p+1} \right] \right),
\end{equation}
where $\kappa_V=\|V\| \cdot \|V^{-1}\|$ is the condition number of $V$.
\end{lemma}

\begin{proof}
The linear multistep method requires a starting method to obtain the values of $\bm{x}_j$ for $0<j<k$.
The term $\max_{0\le j< k} \| \bm{x}(t_j)-\bm{x}_j \|$ arises from the inaccuracy in this procedure.
One can simply use the Euler method, in which case the error is $O(\dt^2)$.
It is also possible to use higher-order starting methods, but there is not a convenient rigorous result that can be used.
To determine the error in the Euler method, one can simply use Theorem \ref{thm2}.
Because $k=1$ and $p=1$ for the Euler method, and for the initial point there is zero error ($\bm{x}(t_0)=\bm{x}_0$), Theorem \ref{thm2} gives
\begin{equation}
\|\bm{x}(t_j)-\bm{x}_j\| \le M_E \kappa_V \dt\int_{t_0}^{t_j}\|\bm{x}^{(2)}(\xi)\|d\xi,
\end{equation}
where $M_E$ is the constant for the Euler method.
Using this expression for $0\le j < k$ yields
\begin{equation}
\max_{0\le j< k}\|\bm{x}(t_j)-\bm{x}_j\| \le M_E \kappa_V \dt^2 (k-1) \max_{\xi\in[t_0,t_0+(k-1)\dt]}\|\bm{x}^{(2)}(\xi)\|.
\end{equation}

In using these results, it is necessary to place upper bounds on the values of $\|\bm{x}^{(p+1)}(\xi)\|$ and $\|\bm{x}^{(2)}(\xi)\|$. In general these will depend on the value of $\bm{b}(t)$, and its time-dependence. It is well-behaved if $\bm{b}$ is a constant, in which case the exact solution is
\begin{equation}
\bm{x}(t) = e^{A(t-t_0)}(\bm{x}_{\rm init} + A^{-1} \bm{b}) - A^{-1} \bm{b}.
\end{equation}
Then
\begin{equation}
\bm{x}^{(\ell)}(t) = e^{A(t-t_0)}(A^\ell \bm{x}_{\rm init} + A^{\ell-1} \bm{b}),
\end{equation}
so
\begin{align}
\|\bm{x}^{(\ell)}(t)\| &= \|V e^{D(t-t_0)} V^{-1}(A^\ell \bm{x}_{\rm init} + A^{\ell-1} \bm{b})\| \nn
&\le \kappa_V (\|A\|^\ell \|\bm{x}_{\rm init}\| + \|A\|^{\ell-1}\|\bm{b}\|)
\end{align}
In the first line we have used the diagonali{\s}ation of $A$, and in the second line we have used the condition that $|\arg(-\lambda_i)|\le \alpha$.
Using Theorem \ref{thm2}, the error is bounded as
\begin{align}
\|\bm{x}(t_m)-\bm{x}_m\| &\le M \kappa_V \left( \max_{0\le j< k} \| \bm{x}(t_j)-\bm{x}_j \| + \dt^p \int_{t_0}^{t_m} \|\bm{x}^{(p+1)}(\xi)\|d\xi\right) \nn
&\le M \kappa_V \left[ M_E \kappa_V^2 \dt^2 k (\|\bm{x}_{\rm init}\| + \|\bm{b}\|/\|A\|)\|A\|^2
 + \dt^p (t_m-t_0) \kappa_V (\|\bm{x}_{\rm init}\| + \|\bm{b}\|/\|A\|)\|A\|^{p+1}\right] \nn
&= O\left( \kappa_V^2 (\|\bm{x}_{\rm init}\| + \|\bm{b}\|/\|A\|)\left[ \kappa_V (\dt \|A\|)^2
 + m(\dt \|A\|)^{p+1}  \right] \right) .
\end{align}
\end{proof}

This result means that, disregarding the dependence on many of the quantities, and omitting the error due to the starting method, the error scales as $O((\|A\|\dt)^p \|A\| \Delta t)$ for total time $\Delta t$. To achieve error bounded by $\epsilon$, we then use
\begin{equation}
\label{eq:nosteps}
N_t = O\left( \frac{(\|A\|\Delta t)^{1+1/p}}{\epsilon^{1/p}} \right).
\end{equation}
That is, the number of time steps required is close to linear in the time.

Now we consider how to encode the multistep method in a linear system $\matr\bm{\vec x} = \bm{\vec b}$.
Earlier we described how this can be done for the Euler method.
More generally, we use the Euler method as a starting method, then continue with a higher-order method, then for the final rows again have $\bm{x}_{j+1}-\bm{x}_j=0$ to ensure that all the final values are equal.
Therefore, we set the blocks of $\matr$ as, for $N_t\ge 2k$,
\begin{equation}
\label{eq:explicit}
\begin{array}{*{20}l}
\matr_{j,j} = \openone, & 0\le j < k, \quad N_t/2 < j \le N_t, \\
\matr_{j,j-1} = -(\openone + A\dt), & 1\le j < k, \\
\matr_{j,j-k+\ell} = \alpha_\ell \openone - \beta_\ell A \dt, & k\le j \le N_t/2, \quad 0\le \ell\le k. \\
\matr_{j,j-1} = -\openone, & N_t/2 < j \le N_t. \\ \end{array}
\end{equation}
We will always require $N_t\ge 2k$ when using $\matr$, because otherwise there are not enough time steps to start the linear multistep method.
We also set the blocks of $\bm{\vec b}$ as
\begin{equation}
\label{eq:explicitb}
\begin{array}{*{20}l}
\bm{b}_{0} = \bm{x}_{\rm in}, & \\
\bm{b}_{j} = \bm{b}\dt, & 1 \le j < k, \\
\bm{b}_{j} = \sum_{\ell=0}^k \beta_\ell \bm{b}\dt, & k\le j \le N_t/2, \\
\bm{b}_{j} = 0, & N_t/2< j \le N_t. \\
\end{array}
\end{equation}

We require $A$, $\bm{b}$, and $\bm{x}_{\rm in}$ to be sparse, with no more than $s$ nonzero elements in any row or column.
We assume that the oracles are of the same form as in Ref.\ \cite{Berry09}.
That is, the oracle for $A$ is a unitary operator acting as
\begin{equation}
O_A \ket{j,\ell}\ket{z} = \ket{j,\ell} \ket{z\oplus A^{[j,\ell]}}.
\end{equation}
Here the $\oplus$ represents modular addition, and $A^{[j,\ell]}$ is given in some binary representation.
Note that the superscript denotes indexing within the block $A$.
We also require an oracle for the sparseness, that locates the nonzero elements. Given a function $f(j,\ell)$ that gives the row index of the $\ell$th nonzero element in column $j$, we require a unitary oracle
\begin{equation}\label{eq:orc}
O_F \ket{j,\ell} = \ket{j,f(j,\ell)}.
\end{equation}
Because $A$ is not Hermitian, we require a similar oracle to give the positions of the nonzero elements in a given row.
We also require oracles to give the values and locations of nonzero elements for $\bm{b}$ and $\bm{x}_{\rm in}$.
These oracles ensure that the initial state corresponding to $\bm{\vec b}$ can be prepared efficiently.
Alternatively, it is also possible to consider $\bm{b}$ and $\bm{x}_{\rm in}$ such the efficient preparation procedure of Ref.\ \cite{Grover02} can be used.

A linear system of equations can be solved using the algorithm of Ref.\ \cite{Harrow09} with complexity $\tilde O(\log(N)s^4\kappa^2/\epsilon_L)$, where $\kappa$ is the condition number of the matrix $\matr$, and $\epsilon_L$ is the allowable error.
(Note that the power of $s$ should be 4, not 2 as given in Ref.\ \cite{Harrow09}.)
Recall that $\epsilon_L$ indicates the allowable error for the solution of the linear systems, which is distinct from the allowable error for the solution of the differential equation, $\epsilon$.

\section{Bounding the condition number}
\label{sec:cond}
To determine the complexity it is necessary to determine the value of the condition number $\kappa$.
To bound this condition number we first determine bounds on the norms of $\|\matr\|$ and $\|\matr^{-1}\|$.

\begin{lemma}
\label{norm1}
The matrix $\matr$, with blocks given by Eq.\ \eqref{eq:explicit}, satisfies $\|\matr\|=O(1)$ provided $\dt=O(1/\|A\|)$.
\end{lemma}

\begin{proof}
To determine the upper bound on $\|\matr\|$, we express $\matr$ as a sum of block-diagonal matrices, and use the triangle inequality. Let us define $\matr^{\{\ell\}}$ to be the block diagonal matrix with all entries zero, except
\begin{equation}
\matr^{\{\ell\}}_{j,j-\ell} = \matr_{j,j-\ell}.
\end{equation}
We then have
\begin{equation}
\matr = \sum_{\ell=0}^k \matr^{\{\ell\}} ,
\end{equation}
so, via the triangle inequality,
\begin{equation}
\|\matr\| \le \sum_{\ell=0}^k \|\matr^{\{\ell\}}\|.
\end{equation}
The norm of a block-diagonal matrix is just the maximum norm of the blocks, so we find
\begin{align}
\|\matr^{\{0\}}\| &\le \max(1,|\alpha_k|+|\beta_k| h \|A\|), \nn
\|\matr^{\{1\}}\| &\le \max(1 + h\|A\|,|\alpha_{k-1}|+|\beta_{k-1}| h \|A\|), \nn
\|\matr^{\{\ell\}}\| &\le |\alpha_{\ell}|+|\beta_{\ell}| h \|A\|, \qquad 1<\ell\le k.
\end{align}
Because we require that $\dt=O(1/\|A\|)$, each of these norms is $O(1)$, and hence the overall norm is $O(1)$.
\end{proof}

\begin{lemma}
\label{norm2}
Suppose that the multistep method is $A(\alpha)$-stable, the matrix $A$ may be diagonalised as $A=VDV^{-1}$, and the eigenvalues of $A$ all satisfy $|\arg(-\lambda_i)|\le \alpha$. Then the matrix $\matr$, with blocks given by Eq.\ \eqref{eq:explicit}, satisfies $\|\matr^{-1}\|=O(N_t\kappa_V)$, where $\kappa_V$ is the condition number of $V$.
\end{lemma}

\begin{proof}
To upper bound $\|\matr^{-1}\|$, we use a method analogous to that used to bound the error in Ref.\ \cite{Hairer}. As in the condition for Theorem \ref{thm2}, we assume that $A$ may be diagonalised as
\begin{equation}
A = V D V^{-1}.
\end{equation}
Note that $A$ need not be Hermitian, so $V$ need not be unitary. If we define ${\cal V}$ to be to the block matrix with $V$ on the diagonal, and ${\cal D}$ to be the matrix corresponding to $\matr$ except with $A$ replaced with $D$, then $\matr={\cal V}{\cal D}{\cal V}^{-1}$. We obtain
\begin{equation}
\|\matr^{-1}\|\le \|{\cal V}\|\cdot\|{\cal D}^{-1}\|\cdot\|{\cal V}^{-1}\| = \kappa_V \|{\cal D}^{-1}\|.
\end{equation}
To bound $\|\matr^{-1}\|$ we therefore just need to bound $\|{\cal D}^{-1}\|$.

The matrix ${\cal D}$ corresponds to the linear multistep solution of decoupled scalar differential equations.
That is, taking $\bm{z}=V^{-1}\bm{x}$, the differential equation becomes $N_x$ decoupled differential equations
\begin{equation}
\dot z^{[j]}(t) = \lambda_j z^{[j]}(t) + [V^{-1}\bm{b}]^{[j]}.
\end{equation}
The matrix ${\cal D}$ gives decoupled linear multistep solutions of each of these differential equations. It may therefore be written in block-diagonal form, with each block corresponding to solution of each of these decoupled equations. The value of $\|{\cal D}^{-1}\|$ can therefore be bounded by the maximum of the norm of the inverse of each of these blocks.

To bound the norm of the inverse, we can take ${\cal D}\bm{\vec z} = \bm{\vec y}$, and determine a bound on the norm of $\bm{\vec z}$ for a given norm of $\bm{\vec y}$. We can determine this by separately examining the uncoupled blocks in ${\cal D}$. For each of these blocks (labelled by $j$) we have the linear multistep equation, for $m=0,\ldots,N_t/2-k$,
\begin{equation}
\label{eq:multier}
\sum_{i=0}^k (\alpha_i - \dt \lambda_j \beta_i) z_{m+i}^{[j]} = y_{m+k}^{[j]}.
\end{equation}
We also have, for the initial condition, $z_{0}^{[j]} = y_{0}^{[j]}$,
and for the Euler method as the starting method with $0\le m<k-1$,
\begin{equation}
z_{m+1}^{[j]} - (1 + \dt \lambda_j) z_{m}^{[j]} = y_{m+1}^{[j]}.
\end{equation}
For the end of the simulation, we have for $N_t/2\le m<N_t$,
\begin{equation}
z_{m+1}^{[j]} - z_{m}^{[j]} = y_{m+1}^{[j]}.
\end{equation}

We can see that Eq.\ \eqref{eq:multier} is equivalent to Eq.\ (7.11) in the method used to bound error in Ref.\ \cite{Hairer}. We identify $z_{m+i}^{[j]}$ as equivalent to $e_{m+i}$ in Ref.\ \cite{Hairer}, and $y_{m+k}^{[j]}$ as equivalent to $\delta_h(x_m)$ in Ref.\ \cite{Hairer}.
As in that method, we can define
\begin{align}
\bm{E}_m &:= (z_{m+k-1}^{[j]},\ldots,z_{m+1}^{[j]},z_{m}^{[j]})^T, \nn
\bm{\Delta}_m &:= (y_{m+k}^{[j]}/(\alpha_k-\dt \lambda_j\beta_k),0,\ldots,0)^T.
\end{align}
As the problem is equivalent to that considered in Ref.\ \cite{Hairer}, the result given in Eq.\ (7.24) of that reference holds:
\begin{equation}
\|\bm{E}_{m+1}\| \le M\left(\|\bm{E}_0\|+\sum_{\ell=0}^m \|\bm{\Delta}_\ell\| \right),
\end{equation}
where $M$ is a constant depending only on the method. Using the definition of $\bm{\Delta}_\ell$ gives
\begin{align}
\|\bm{E}_{m+1}\| &\le M\left(\|\bm{E}_0\|+\sum_{\ell=0}^{m} |y_{\ell+k}^{[j]}|/|\alpha_k-\dt \lambda_j\beta_k| \right) \nn
&\le M\left(\|\bm{E}_0\|+\sum_{\ell=0}^{m} |y_{\ell+k}^{[j]}|/|\alpha_k| \right).
\end{align}
In the last line we have used the fact that the condition of $A(\alpha)$ stability means $\alpha_k\cdot\beta_k>0$, so $|\alpha_k - \dt \lambda_j\beta_k|^{-1}\le |\alpha_k|^{-1}$.

For the starting method, we have used the Euler method, and the result is simpler.
For the Euler method, $\bm{E}_m$ and $\bm{\Delta}_m$ are scalars, and are just $z_m^{[j]}$ and $y_{m+1}^{[j]}$.
The corresponding result is therefore, for $0< m < k$,
\begin{align}
|z_m^{[j]}| &\le M_E\left(|z_0^{[j]}|+\sum_{\ell=0}^{m-1} |y_{\ell+1}^{[j]}| \right) \nn
&= M_E\sum_{\ell=0}^{m} |y_{\ell}^{[j]}|.
\end{align}
Here $M_E$ is the corresponding constant for the Euler method.
For the end of the simulation, we have for $N_t/2\le m<N_t$, $z_{m+1}^{[j]} - z_{m}^{[j]} = y_{m+1}^{[j]}$, so
\begin{equation}
|z_m^{[j]}| \le |z_{N_t/2}^{[j]}| + \sum_{\ell=N_t/2+1}^m | y_\ell^{[j]} |.
\end{equation}

The norm of $\bm{E}_0$ can be bound as
\begin{align}
\|\bm{E}_0\| &\le \sqrt{\sum_{m=0}^{k-1} |z_m^{[j]}|^2} \nn
&\le M_E \sqrt{k} \sum_{\ell=0}^{k-1} |y_{\ell}^{[j]}|.
\end{align}
Using this result, $|z_{m}^{[j]}|$ can be bound as, for $k\le m \le N_t/2$,
\begin{align}
|z_{m}^{[j]}| &\le \| \bm{E}_{m+1-k}\| \nn
& \le M\left(\|\bm{E}_0\|+\sum_{\ell=0}^{m-k} |y_{\ell+k}^{[j]}|/|\alpha_k| \right) \nn
& \le M\left(M_E \sqrt{k} \sum_{\ell=0}^{k-1} |y_{\ell}^{[j]}|+\sum_{\ell=k}^{m} |y_{\ell}^{[j]}|/|\alpha_k| \right) .
\end{align}

For convenience we define the quantity
\begin{equation}
M_T := \max(MM_E\sqrt{k},M/|\alpha_k|,M_E,1).
\end{equation}
Then we find that, for all $0\le m\le N_t$,
\begin{equation}
|z_{m}^{[j]}| \le M_T \sum_{\ell=0}^{m} |y_{\ell}^{[j]}|.
\end{equation}
Hence we can determine an overall upper bound on the norm of $\vec z^{[j]}$ as
\begin{align}
\|\vec{z}^{[j]}\|^2 &\le M_T^2\sum_{m=0}^{N_t} \left(\sum_{\ell=0}^{m} |y_{\ell}^{[j]}|\right)^2 \nn
&\le M_T^2 N_t^2 \|\vec{y}^{[j]}\|^2. 
\end{align}
Summing over $j$, this then gives
\begin{align}
\|\bm{\vec z}\| \le M_T N_t \|\bm{\vec y}\|. 
\end{align}
This result means that $\|{\cal D}^{-1}\|\le M_TN_t$, where $M_T$ depends only on the method. This then bounds the norm of $\matr^{-1}$ as
\begin{equation}
\|\matr^{-1}\| = O(N_t\kappa_V).
\end{equation}
\end{proof}

We can now use these results to bound the condition number of $\matr$.

\begin{theorem}
\label{thm:conthm}
Suppose that the multistep method is $A(\alpha)$-stable, the matrix $A$ may be diagonalised as $A=VDV^{-1}$, the eigenvalues of $A$ all satisfy $|\arg(-\lambda_i)|\le \alpha$, and $\dt=O(1/\|A\|)$. Then the matrix $\matr$, with blocks given by Eq.\ \eqref{eq:explicit}, has condition number $\kappa=O(N_t\kappa_V)$, where $\kappa_V$ is the condition number of $V$.
\end{theorem}

\begin{proof}
The condition number is given by the formula
\begin{equation}
\kappa = \left(\max_{\bm{\vec x}} \frac{\| \matr\bm{\vec x} \|}{\|\bm{\vec x}\|}\right)
\left(\max_{\bm{\vec x}} \frac{\|\bm{\vec x}\|}{\| \matr\bm{\vec x} \|}\right) = \|\matr\| \cdot \|\matr^{-1}\|.
\end{equation}
The conditions of this theorem ensure that the conditions of Lemmas \ref{norm1} and \ref{norm2} hold. Therefore we can use the bounds on $\|\matr\|$ and $\|\matr^{-1}\|$ from those lemmas to obtain $\kappa=O(N_t \kappa_V)$.
\end{proof}

This result for the condition number can be explained in a more intuitive way.
Each value of $y_\ell^{[j]}$ is equivalent to an excitation of the differential equation at a single time.
Therefore $\bm{\vec z}$ is close to the solution of the differential equation with each of those individual excitations.
An excitation can not cause a growing solution, because of the stability condition.
This means that the worst that an excitation can do is cause a solution that is displaced by a proportional amount for the remaining time.
Therefore the norm of $\bm{\vec z}$ can not be more than a factor of $N_t$ times the norm of $\bm{\vec y}$.

\section{Algorithm for solving linear systems}
\label{sec:fullalg}
Next the bound on the condition number will be used to estimate the complexity of the quantum algorithm.
We will first explain the scaling in a simple way, then give a more rigorous result.
Using Eq.\ \eqref{eq:nosteps} for the number of time steps, we have (ignoring dependence on many of the quantities)
\begin{equation}
\kappa = O\left( \frac{(\|A\|\Delta t)^{1+1/p}}{\epsilon^{1/p}} \right).
\end{equation}
Using this expression in the result for the complexity of solving linear systems from Ref.\ \cite{Harrow09} gives
\begin{equation}
\tilde O(\log(N)s^4(\|A\|\Delta t)^{2+2/p}/(\epsilon^{2/p}\epsilon_L)).
\end{equation}

Before obtaining the overall scaling, another factor that needs to be considered is the scaling for creating the state encoding $\bm{\vec b}$.
This is because the algorithm of Ref.\ \cite{Harrow09} uses an amplitude amplification approach, which requires the preparation of this state at each step.
Therefore, the overall complexity is multiplied by the complexity of performing this state preparation.
We assume that we have oracles that give the elements of $\bm{x}_{\rm in}$ and $\bm{b}$ in the computational basis, and that these vectors are $s$-sparse.
We find the following result.

\begin{lemma}
\label{lem:prep}
The state encoding $\bm{\vec b}$,
\begin{equation}
\ket{\bm{\vec b}} \propto \sum_{j,\ell} b_j^{[\ell]} \ket{j,\ell},
\end{equation}
can be prepared using $O(\sqrt{s}+\log(N_t))$ calls to the oracles for $\bm{x}_{\rm in}$ and $\bm{b}$, provided the normali{\s}ations of $\bm{x}_{\rm in}$ and $\bm{b}$ are known.
\end{lemma}

\begin{proof}
For this preparation we can assume that the normali{\s}ations of $\bm{x}_{\rm in}$ and $\bm{b}$ are known.
That is because this normali{\s}ation can be determined with complexity $O(s)$, and need only be determined once.
Because the overall complexity of the simulation is greater than this, it can just be assumed that the normali{\s}ation is known.

A state of dimension $s$ can be prepared with complexity $O(\sqrt{s})$ using the method of Ref.~\cite{Grover00}.
We take the oracle that gives the positions of the nonzero elements to be of the form given in Eq.~\eqref{eq:orc}.
This form is the same as in Ref.~\cite{Berry09}, and is slightly stronger than the form given in Ref.~\cite{Berry07}.
To obtain an oracle of this form, one simply needs to be able to invert the function for the positions of the nonzero elements.
This inversion can be performed efficiently if the positions are given sequentially.
Using an oracle of this form simplifies the problem, because it only requires one oracle call to prepare an $s$-sparse state from a dimension $s$ state with the same coefficients \cite{Berry09}. Therefore the complexity of preparing an $s$-sparse state is also $O(\sqrt{s})$.

Let us encode the state in three registers (which may themselves be composed of multiple qubits).
The first is one qubit encoding $\ket{0}$ for the times $t_1,\ldots,t_{N_t/2}$, and $\ket{1}$ for the times $t_0,t_{N_t/2+1},\ldots,t_{N_t}$.
The second register provides the remainder of the encoding of the time.
The third register is of dimension $N_x$, and encodes $\bm{b}$ or $\bm{x}_{\rm in}$.

By performing a rotation on the first qubit, based on the normali{\s}ations of $\bm{x}_{\rm in}$ and $\bm{b}$, we obtain the correct relative weighting of the two time ranges.
Then, conditional on the qubit being in the state $\ket{0}$, we can prepare a superposition of times $t_1,\ldots,t_{N_t/2}$ in the second register, as well as a state encoding $b$ in the third register.
Conditional on the qubit being in the state $\ket{1}$, we can prepare the time $t_0$ in the second register, and $\bm{x}_{\rm in}$ in the third register.
The complexity of these controlled state preparations will be $O(\sqrt{s})$.
The complexity of preparing the superposition of times can be made $O(\log(N_t))$ simply by choosing $N_t$ to be a power of two (which does not change the scaling).
\end{proof}

We now translate the complexity of solving linear systems into the complexity of obtaining a state corresponding to the solution of the differential equation.

\begin{theorem} \textbf{(Main result)}
\label{thm:final}
Suppose that the multistep method is order $p$ and $A(\alpha)$-stable, the matrix $A$ may be diagonalised as $A=VDV^{-1}$, the eigenvalues of $A$ all satisfy $|\arg(-\lambda_i)|\le\alpha$, and
\begin{align}
\label{eq:varcon}
\max_{t\in[t_0,t_0+\Delta t]}\|\bm{x}(t)\| &= O(\|\bm{x}(t_0+\Delta t)\|), \\
\label{eq:xfbig}
\epsilon &= o(\|\bm{x}_{\rm in}\|).
\end{align}
Then a state encoding the solution of the differential equation at time $t_0+\Delta t$ to within trace distance $\epsilon$ may be obtained using
\begin{equation}
\label{eq:main}
\tilde O\left(\log(N_x)s^{9/2} (\|A\|\Delta t)^{2+2/p}\kappa_V^{5}(\|\bm{x}_{\rm in}\|+\|\bm{b}\|/\|A\|)/\epsilon^{2}\right)
\end{equation}
calls to the oracles for $A$, $\bm{b}$, and $\bm{x}_{\rm in}$.
\end{theorem}

\begin{proof}
There are two main issues that need to taken into account in determining the complexity of obtaining the solution of the differential equation.
The first is the relation between the allowable error for solving the differential equation and the allowable error for the solution of the linear systems,
and the second is the probability of obtaining the correct time.

The error in the algorithm for the solution of linear systems \cite{Harrow09} is the error in the state produced.
Denoting the exact desired \emph{normali{\s}ed} (indicated by the tilde) state by $\ket{\tilde\psi}$ and the actual state produced by $\rho$, we therefore have
\begin{equation}
\delta(\rho,\ket{\tilde\psi}\bra{\tilde\psi})\le \epsilon_L,
\end{equation}
where $\delta$ denotes the trace distance.
The trace distance is non-increasing under channels, so the channel that (non-destructively) measures the time does not increase the trace distance.
Hence
\begin{align}\label{eq:dist}
\epsilon_L & \ge \delta\left(\sum_{j=0}^{N_t}\ket{t_j}\bra{t_j}\bra{t_j}\rho\ket{t_j},\sum_{i=0}^{N_t}\ket{t_j}\bra{t_j}\braket{t_j}{\tilde\psi}\braket{\tilde\psi}{t_j}\right) \nn
& = \sum_{j=0}^{N_t} \delta\left(\bra{t_j}\rho\ket{t_j},\frac{\ket{\bm{x}_j}\bra{\bm{x}_j}}{\braket{\psi}{\psi}}\right) \nn
& \ge  \delta\left(\sum_{j=N_t/2}^{N_t}\bra{t_j}\rho\ket{t_j},\frac{(N_t/2+1) \ket{\bm{x}_{N_t}}\bra{\bm{x}_{N_t}}}{\braket{\psi}{\psi}}\right).
\end{align}
To obtain the correct final state, one would need to measure the time in the ancilla register, and obtain a value in the range $t_0+\Delta t$ to $t_0+2\Delta t$.
This is equivalent to obtaining $j$ in the range $N_t/2$ to $N_t$.

The probability of successfully obtaining the time in the correct range is given by
\begin{equation}
p_{\rm succ} :=\sum_{j=N_t/2}^{N_t}{\rm Tr}(\bra{t_j}\rho\ket{t_j}).
\end{equation}
Then the state we obtain after measuring the time in the correct range is
\begin{equation}
\rho_{\rm post} =  p_{\rm succ}^{-1}
\sum_{j=N_t/2}^{N_t}\bra{t_j}\rho\ket{t_j}.
\end{equation}
It is also convenient to define
\begin{equation}
p_{\rm ex} :=\frac{(N_t/2+1) \braket{\bm{x}_{N_t}}{\bm{x}_{N_t}}}{\braket{\psi}{\psi}}.
\end{equation}
This is the probability of success if the linear system was solved exactly.
Then Eq.~\eqref{eq:dist} yields
\begin{equation}
\delta\left(p_{\rm succ}, p_{\rm ex} \right) \le \epsilon_L ,
\end{equation}
and hence
\begin{equation}
\delta\left(p_{\rm succ}p_{\rm ex}^{-1},1\right) \le \epsilon_L p_{\rm ex}^{-1}.
\end{equation}

Next, using the triangle inequality for trace distance, we obtain
\begin{align}\label{eq:trbnd}
\delta\left(\rho_{\rm post} , \ket{\tilde{\bm{x}}_{N_t}}\bra{\tilde{\bm{x}}_{N_t}}\right) &\le
\delta\left( p_{\rm succ}p_{\rm ex}^{-1},1\right)
+\delta\left(p_{\rm ex}^{-1}\sum_{j=N_t/2}^{N_t}\bra{t_j}\rho\ket{t_j},\ket{\tilde{\bm{x}}_{N_t}}\bra{\tilde{\bm{x}}_{N_t}}\right) \nn
&= O\left(\epsilon_L p_{\rm ex}^{-1}\right),
\end{align}
where $\ket{\tilde{\bm{x}}_{N_t}}$ is the normali{\s}ed form of $\ket{\bm{x}_{N_t}}$.
Next we consider how to place a lower bound on $p_{\rm ex}$.
The normali{\s}ation of the state $\ket{\psi}$ is given by
\begin{align}
\label{eq:norm}
\braket{\psi}{\psi} &= \sum_{j=0}^{N_t} \|\bm{x}_j\|^2 \nn
&= \sum_{j=0}^{N_t} [\|\bm{x}(t_0+\dt j)\|+O(\epsilon)]^2 \nn
&= O\left(N_t\max_{t\in[t_0,t_0+\Delta t]}\|\bm{x}(t)\|^2\right) \nn
&= O\left(N_t\|\bm{x}(t_0+\Delta t)\|^2\right).
\end{align}
Here we have bounded the error in $\bm{x}(t)$ all times by $\epsilon$.
This is because we choose parameters that bound the error at time $t_0+\Delta t$ by $\epsilon$.
The bound on the error increases monotonically with time, so the error at earlier times will also be bounded by $\epsilon$.
The conditions \eqref{eq:varcon} and \eqref{eq:xfbig} ensure that $\epsilon=o(\|\bm{x}(t_0+\Delta t)\|)$.
Using this result for $\braket{\psi}{\psi}$ we obtain
\begin{equation}
p_{\rm ex} = \Omega \left( \frac{\braket{\bm{x}_{N_t}}{\bm{x}_{N_t}}}{\|\bm{x}(t_0+\Delta t)\|^2}\right) = \Omega(1).
\end{equation}
As a result, Eq.~\eqref{eq:trbnd} gives
\begin{align}
\delta\left(\rho_{\rm post} , \ket{\tilde{\bm{x}}_{N_t}}\bra{\tilde{\bm{x}}_{N_t}}\right) = O\left(\epsilon_L \right).
\end{align}
To obtain trace distance error $O(\epsilon)$, we can therefore take $\epsilon_L = \Theta(\epsilon)$.
Moreover, because $\delta\left(p_{\rm succ}, p_{\rm ex} \right) \le \epsilon_L$, the probability of success is also $\Omega(1)$.

Using $\epsilon_L = \Theta(\epsilon)$  in the scaling from Ref.\ \cite{Harrow09}, together with $\kappa=O(N_t \kappa_V)$ from Theorem \ref{thm:conthm}, and the complexity of state preparation from Lemma \ref{lem:prep}, the number of oracle calls is
\begin{equation}
\tilde O(\log(N_x)s^{9/2} N_t^{2} \kappa_V^2 /\epsilon).
\end{equation}
Note that we can omit $\log(N_t)$ from Lemma \ref{lem:prep}, because the $\tilde O$ notation omits logarithmic factors. For the same reason, we have replaced $N=N_x N_t$ with $N_x$.
We use a value of $N_t$ that is sufficient to ensure that the error is no greater than $\epsilon$, and that $\dt=O(1/\|A\|)$, which is a condition needed to use Theorem \ref{thm:conthm}. If we take
\begin{equation}
N_t = \Theta \left( \|A\|\Delta t \sqrt{\frac{\kappa_V^3(\|\bm{x}_{\rm in}\|+\|\bm{b}\|/\|A\|)}{\epsilon}}
+ (\|A\|\Delta t)^{1+1/p} \left( \frac{\kappa_V^2(\|\bm{x}_{\rm in}\|+\|\bm{b}\|/\|A\|)}{\epsilon} \right)^{1/p} \right),
\end{equation}
then using Lemma \ref{lem:ersca}, the error will be bounded by $\epsilon$.
In addition, because $\epsilon=o(\|\bm{x}_{\rm in}\|)$, we obtain $N_t = \Omega (\|A\|\Delta t)$, which ensures that $\dt=O(1/\|A\|)$.

We can simplify the result by taking
\begin{equation}
N_t = \Theta \left( (\|A\|\Delta t)^{1+1/p} \sqrt{\frac{\kappa_V^3(\|\bm{x}_{\rm in}\|+\|\bm{b}\|/\|A\|)}{\epsilon}}
\right).
\end{equation}
Then the overall scaling of the number of black-box calls is as in Eq.~\eqref{eq:main}.
\end{proof}

This result is somewhat conservative, because we have included the term due to the error in starting the linear multistep method. If we assume that this error is negligible, then we obtain
\begin{equation}
\tilde O\left(\log(N_x)s^{9/2} (\|A\|\Delta t)^{2+2/p}\kappa_V^{2+4/p}(\|\bm{x}_{\rm in}\|+\|\bm{b}\|/\|A\|)^{1/p}/\epsilon^{1+2/p}\right).
\end{equation}
This has the same scaling in $\|A\|\Delta t$, but improved scaling in other quantities.

\section{Variable time amplitude amplification}
\label{sec:var}
The algorithm that has been presented scales close to $\Delta t^{2}$; it is likely that this is suboptimal, because the lower bound is linear scaling.
The factor that increases the scaling with $\Delta t$ is that the HHL algorithm scales as $\kappa^2$, whereas the lower bound is linear scaling in $\kappa$.
A proposed technique to improve the scaling to close to linear in $\kappa$ was given in \cite{Ambainis10b,Ambainis12}.
There the scaling was (see Theorem 3 of \cite{Ambainis10b})
\begin{equation}
O\left( \frac{\kappa \log^3\frac{\kappa}{\epsilon_L}}{\epsilon_L^3} \log^2 \frac 1{\epsilon_L}\right).
\end{equation}
Here only scaling in $\kappa$ and $\epsilon_L$ has been included.
As above, $\kappa$ scales linearly in $N_t$, which is close to linear in $\Delta t$.
This would appear to yield an algorithm with scaling close to linear in $\Delta t$.

However, there is an additional contribution to the complexity from the \textbf{Estimate} procedure that prevents this scaling being obtained. 
In Algorithm 3 of Ref.~\cite{Ambainis10b} the \textbf{Estimate} procedure is used to determine $p_i$ from algorithm $\B_i$.
That is then used to determine the value of $k$ to use in Algorithm 2 to define algorithm $\A_i$.
Then $\A_i$ is used to define $\B_{i+1}$ and so forth in a recursive way.

To determine the complexity for this recursive procedure, one needs to find the complexity for $\A_i$, denoted $T_i$, in terms of $T_{i-1}$.
In point 1 of the proof of Lemma 2 of Ref.~\cite{Ambainis10b} the complexity due to amplitude amplification is taken account of, but not the complexity of the \textbf{Estimate} procedure.
To obtain a bound on the contribution of the complexity of the \textbf{Estimate} procedure, note that the complexity in Theorem 4 of Ref.~\cite{Ambainis10b} scales at least as $1/\sqrt{p}$, and from Algorithm 2, $1/p$ should be at least $m$.
Multiplying the factor of $\sqrt{m}$ by itself $m$ times to account for the number of steps of the recursion yields $m^{m/2}$.
As $m=\log T_{\rm max}$, $m^{m/2}$ scales as $T_{\rm max}^{(\log \log T_{\rm max})/2}$.
As used for solving linear systems, $T_{\rm max} =O( \kappa/\epsilon_L$), which would yield scaling as $(\kappa/\epsilon_L)^{[\log\log(\kappa/\epsilon_L)]/2}$.
In the limit of large $\kappa/\epsilon_L$, the double log becomes large, and the scaling is significantly larger than linear in $\kappa/\epsilon_L$.
As a result it appears that variable time amplitude amplification does not yield improved scaling in $\Delta t$.

\section{Conclusions}
A quantum computer may be used to solve sparse systems of linear differential equations, provided the result may be encoded in a quantum state, rather than given explicitly. By encoding the differential equation as a linear system, and using the algorithm of Ref.\ \cite{Harrow09} for solving linear systems, the complexity is (including only scaling in $\|A\|$ and $\Delta t$),
\begin{equation}
\tilde O\left((\|A\|\Delta t)^{2}\right).
\end{equation}
This improves upon previous results for nonlinear differential equations, which were exponential in the number of time steps \cite{Leyton08}. This algorithm has an enormous range of possible applications, because large systems of differential equations are ubiquitous in science and engineering.  In particular they arise from the discreti{\s}ation of partial differential equations.
There is an interesting open question of whether the scaling can be improved, because the best known lower bound is $O(\|A\|\Delta t)$ from the no--fast-forwarding theorem.

These results are for constant coefficients, because that enables an analytic error analysis. This approach can also be used to solve linear differential equations with time-dependent coefficients, though the error analysis will be more difficult.
Another interesting direction for future research is the problem of nonlinear differential equations. It is likely that the exponential scaling obtained in Ref.\ \cite{Leyton08} is fundamental, because quantum mechanics is linear. However, it may potentially be possible to improve the constant of the exponential scaling.

\acknowledgements
The author is grateful for enlightening discussions with Andrew Childs, Barry Sanders, Ish Dhand, and Jon Tyson.
The author is funded by an Australian Research Council Future Fellowship (FT100100761).

\end{document}